\documentclass[%
 amsmath,amssymb,
 aps,
 twocolumn,
pra,
floatfix,
]{revtex4-2}

\usepackage{graphicx}
\graphicspath{{figures/}}
\usepackage{caption}
\usepackage{subcaption}
\usepackage{dcolumn}
\usepackage{bm}
\usepackage{multirow}
\usepackage{xcolor}
\usepackage{hyperref}

\usepackage{amsthm}
\newtheorem{theorem}{Theorem}
\newtheorem{lemma}{Lemma}
\newcommand{\tr}{\operatorname{Tr}}

\begin{document}


\title{Benchmarking Quantum Processor Performance at Scale}

\author{David C. McKay}
\email{dcmckay@us.ibm.com}
\author{Ian Hincks}
\author{Emily J. Pritchett}
\author{Malcolm Carroll}
\author{Luke C. G. Govia}
\author{Seth T. Merkel}
\affiliation{IBM Quantum}

\date{\today}

\begin{abstract}
As quantum processors grow, new performance benchmarks are required to capture the full quality of the devices at scale. While quantum volume is an excellent benchmark, it focuses on the highest quality subset of the device and so is unable to indicate the average performance over a large number of connected qubits. Furthermore, it is a discrete pass/fail and so is not reflective of continuous improvements in hardware nor does it provide quantitative direction to large-scale algorithms. For example, there may be value in error mitigated Hamiltonian simulation at scale with devices unable to pass strict quantum volume tests. Here we discuss a scalable benchmark which measures the fidelity of a connecting set of two-qubit gates over $N$ qubits by measuring gate errors using simultaneous direct randomized benchmarking in disjoint layers. Our layer fidelity can be easily related to algorithmic run time, via $\gamma$ defined in Ref.~\cite{berg2022probabilistic} that can be used to estimate the number of circuits required for error mitigation. The protocol is efficient and obtains all the pair rates in the layered structure. Compared to regular (isolated) RB this approach is sensitive to crosstalk. As an example we measure a $N=80~(100)$ qubit layer fidelity on a 127 qubit fixed-coupling ``Eagle'' processor (ibm\_sherbrooke) of 0.26(0.19) and on the 133 qubit tunable-coupling ``Heron'' processor  (ibm\_montecarlo) of 0.61(0.26). This can easily be expressed as a layer size independent quantity, error per layered gate (EPLG), which is here $1.7\times10^{-2}(1.7\times10^{-2})$ for ibm\_sherbrooke and $6.2\times10^{-3}(1.2\times10^{-2})$ for ibm\_montecarlo. 
\end{abstract}

\maketitle


The development of quantum benchmarks enables improvements to be tracked across devices and technologies so that reasonable inferences on performance can be made. In Ref.~\cite{amico:2023}, some properties of quantum benchmarks were discussed, and that a suite of benchmarks should be designed to address quality, speed and scale, altogether describing performance. There are few suggested speed benchmarks besides CLOPS~\cite{wack2021quality}; however, for quality and scale there are many proposals. Generally, the quality is signified by having high fidelity gates (the underlying operations of the device) over a large set of connected qubits with low crosstalk. The size of the set is a benchmark of scale. Such quality can be measured discretely, by individually benchmarking the gate, or holistically, e.g.,~by running large representative circuits with well known outputs. 

Individual gate quality is typically measured by variants of randomized benchmarking~\cite{magesan:2012,Helsen22} (RB). For RB, one selects a random sequence of Clifford gates, constructs a circuit by appending the inverse of the sequence (also a Clifford), decomposes this circuit into the native gate set of a device, and then runs the circuit on said device. The decay of the measured polarization (of any Pauli-$Z$ operator) versus sequence length averaged over many random sequences is straightforwardly related to the average gate error. Because these sequence lengths can be very deep, small errors can be measured that are not dependent on state-preparation and measurement (unlike tomography). Measured in this way, we obtain fine-grained information about the device since we have error rates on each discrete gate element. However, important features of the noise can be missed depending on the way these are measured. Specifically, in a connected device, if we measure isolated two-qubit (2Q) gate pairs using RB, we potentially overlook crosstalk terms. This issue was addressed in the simultaneous RB protocol~\cite{gambetta:2012}, yet there are still ambiguities in the implementation. 

Conversely, running test algorithms/structured circuits can give a holistic view of gate quality; however, it is very specific to the type of circuits selected. There have been proposed families of circuits as benchmarks \cite{lubinski2021application, mesman2021qpack, finvzgar2022quark, tomesh2022supermarq, zhang2023characterizing, kurlej2022benchmarking, lubinski2023optimization, kordzanganeh2022benchmarking, mundada2022experimental,li:2022}; however, it remains an open question how to connect performance on one such benchmark to another. As such, benchmarks based on randomized circuits, such as quantum volume (QV)~\cite{cross2019validating}, cross-entropy benchmarking (XEB)~\cite{boixo:2018}, mirror RB~\cite{proctor:2022}, and inverse-free (binary) RB~\cite{hines:2023}, are often believed to give a better overview of average performance. In particular, QV is a stringent test of the device which is defined as QV$=2^N$ when some subset of $N$ qubits can pass the QV test -- to measure a heavy output probability greater than $2/3^{\mathrm{rd}}$ for circuits with $N$ random, all-to-all connected, SU(4) layers. It is straightforward to compare QV across different qubit technologies, and its performance has been linked to the performance of quantum error correcting codes~\cite{baldwin:2022}. 

However, as with all benchmarks, there are limits to QV. For one, it reports on the performance of the best subset of qubits on a device; it is a “high-flier” benchmark. For devices with more qubits than log$_2$(QV), QV is not a good representative number of overall quality. For example, in superconducting qubits the largest quantum volume is 512 (9 qubits)~\cite{ibmqv512} and in ion traps 524288 (19 qubits)~\cite{quantinuum_qv}. Yet, there are devices being constructed at scales far beyond these thresholds and so QV is not capturing quality across the full scale of the device. Secondly, QV (and also XEB) requires classical computation of the circuits, and so is limited to scales where that is tractable -- generally thought to be about 50 qubits and 50 layers of gates (see recent review Ref.~\cite{xiaosi:2023} and $53\times20$ simulation on advanced high-performance computing (HPC)~\cite{yong:2022}). And thirdly, as a discrete benchmark, QV does not indicate continuous changes in gate improvement. Finally, QV measures a specific type of unstructured square circuit; however, many near term algorithms, such as the variational quantum eigensolver (VQE), quantum approximate optimization algorithm (QAOA), and Trotterized dynamics~(see, e.g., Ref~\cite{dalzell:2023} for a review), are based on the idea of a repetitive layer of gates. Similarly, quantum error correction (QEC) relies on a repetitive structure of the application of parallel gates and measurements to perform code checks and detect errors (see, e.g., Ref.~\cite{liepelt:2023} for a QEC inspired benchmark). 

Layered circuits lend themselves well to applying the techniques of error migitation~\cite{berg2022probabilistic, ferracin:2022}.  Error mitigation is a post processing technique that makes a tradeoff with speed to improve quality, i.e., by running more instances of the circuit with different noise profiles to purify the final result. Therefore, there are compelling reasons to provide a benchmark that spans across an entire device via layered circuits and which reveals continuous information as a complement to QV. While XEB, mirror RB and binary RB can probe layered circuits, they require high-weight measurements that do not reveal information about individual gates. Furthermore, XEB has similar classical computational limitations as QV and the output fidelity can be optimized over any N-qubit unitary in each layer. This adds flexibility to XEB, but makes device to device and application to application comparisons difficult. 

To address these points, we propose an alternative benchmark called layer fidelity ($LF$), which combines the ideas of simultaneous~\cite{gambetta:2012} and direct~\cite{proctor:2019} randomized benchmarking and is summarized graphically in Fig.~\ref{fig:lf_schematic}. For a given fully connected set of 2Q gates, we partition them into $M$ layers where the 2Q gates are disjoint. When in disjoint layers, we can construct simultaneous direct randomized benchmarking sequences for these gates with alignment barriers and measure individual 1Q and 2Q fidelities. From these disjoint fidelities we can use the product to estimate the full layer fidelity over $N$ qubits. Given this measurement we have enough information to estimate the layer fidelity of all embedded layers of size $<N$.  To normalize to a size-independent quantity, we introduce error per layered gate (EPLG), $\mathrm{EPLG}=1-LF^{1/n_{2Q}}$ where $n_{2Q}$ is the number of two-qubit gates (typically $N-1$ for a linear chain of qubits), which is representative of the process error of a gate in these layered circuits. A similar quantity, the dressed two-qubit pauli error (measured from XEB), was defined in Ref.~\cite{morvan:2023}. Lending support to the $LF$, Ref.~\cite{morvan:2023} shows a threshold between a quantity similar to $LF$ and the ability to classically simulate random circuits with layered structure. 

We will discuss the full algorithm in \S~\ref{sect:protocol} and show data on two IBM devices (127 qubit and 133 qubit) in \S~\ref{sect:data}. In contrast to other protocols that Pauli-twirl a repeated layer~\cite{berg2022probabilistic,erhard:2019,helsen:2019,kimmel:2014,carignan-dugas:2023}, the procedure for calculating $LF$ requires fewer circuits. However, we can still relate $LF$ to a mitigation metric under most conditions, $\gamma=1/LF^2$; $\gamma$ links noise to the number of probabalistic error cancellation circuits~\cite{berg2022probabilistic} required for a depth $\delta$ circuit, $O(\gamma^{2 \delta})$~\footnote{$\delta$ here is the number of repeated full layers which are used to compute $\gamma$. If $\delta$ is the traditionally defined circuit depth, $\gamma$ is a geometric mean over the disjoint layers and is computed from $LF$ as such.}. We show data comparing LF and $\gamma$ in \S~\ref{sect:add_data} and further discussion of the bounds is in \S~\ref{sect:gamma}. $LF$ is similarly linked to the quantity measured by mirror RB~\cite{proctor:2022}; we show data comparing $LF$ and mirror RB in \S~\ref{sect:add_data}, and we compare simulations of RB, $LF$ and mirror in \S~\ref{sect:sims}.

\section{Layer Fidelity Protocol \label{sect:protocol}}

\begin{figure*}[ht!]
     \centering
         \includegraphics[width=0.85\textwidth]{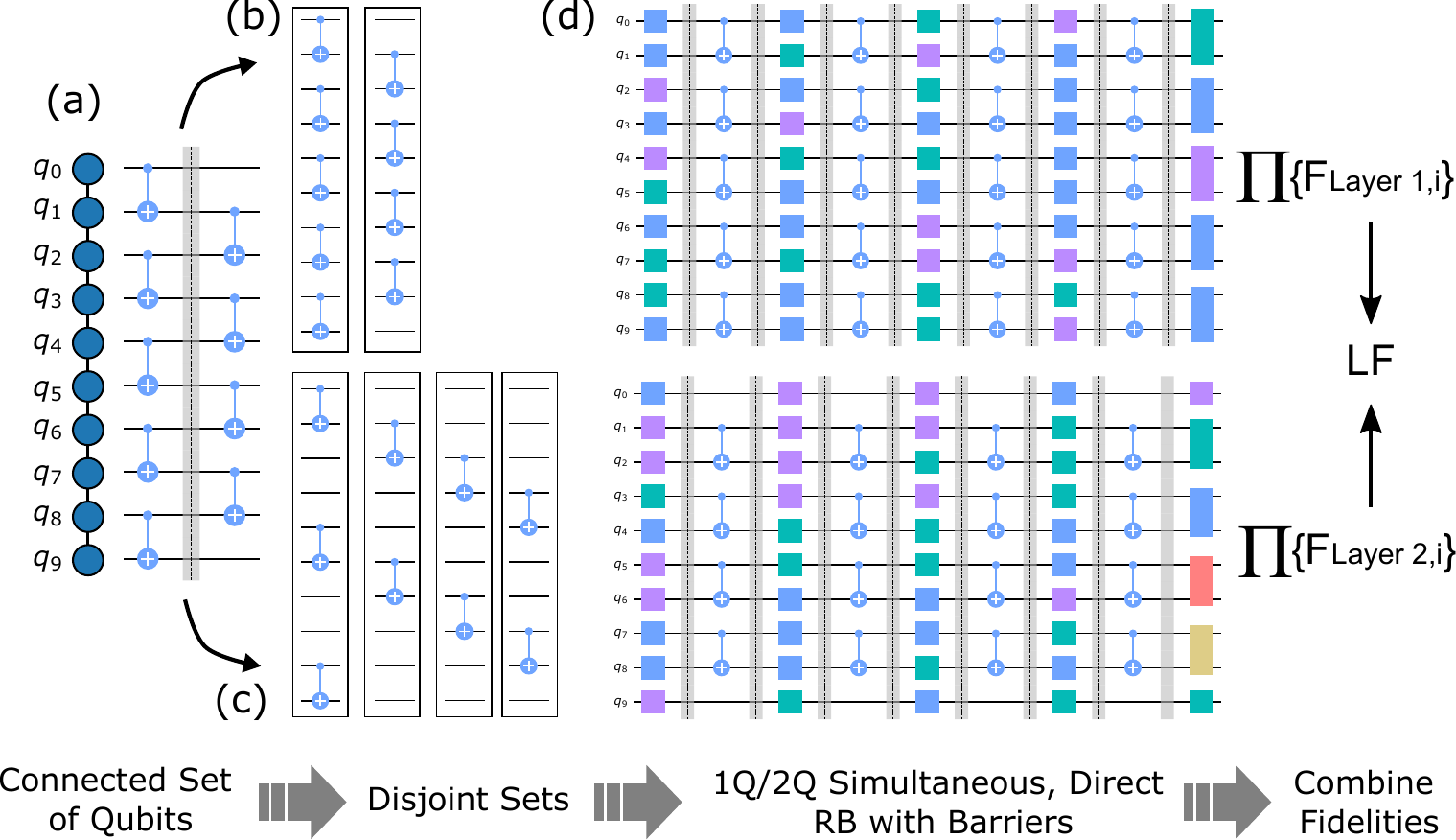}
     \hfill
        \caption{(a) Here we consider a linear chain of qubits with nearest neighbor coupling for which a connecting set of gates is comprised of a disjoint layer of gates starting on qubit 0 followed by a disjoint layer of gates starting on qubit 1. The disjoint layers (b) can either be the maximally simultaneous sets, but could alternatively be a more sparse set (c) split into more disjoint layers. (d) For the disjoint layer set of (b) this requires two simultaneous direct RB experiments here shown for depth $l=4$ with the last layer the inverses in each disjoint space. We measure decay curves as a function of $l$  and fit to extract the fidelities, which are then multiplied together as given by Eqn.~\ref{eqn:lf} to obtain the layer fidelity. }
        \label{fig:lf_schematic}
\end{figure*}

An overview of the protocol is visualized in Fig.~\ref{fig:lf_schematic}, and here we outline the steps of the protocol,
\begin{enumerate}
\item Select a set of $N$ qubits $\{q_i\}$ with a connected set of Clifford two qubit gates $U=\{U_{ij}\}$ (e.g., CNOT) such that the set of two qubit gates plus arbitrary single-qubit gates define a universal gate set over $\{q_i\}$. Part (a) of Fig.~\ref{fig:lf_schematic}.
\item Split the full layer into $M$ disjoint layers with $\{U_{ij}\}_m$ such that $U=\sum_{m}^{M} \{U_{ij}\}_m $ where $\{U_{ij}\}_m$ have no overlapping qubits. The set of idle qubits are $\{q_i\}_m$. Example disjoint layers shown in (b) and (c) of Fig.~\ref{fig:lf_schematic}. 
\item Measure the errors on $\{U_{ij}\}_m$ and $\{q_i\}_m$ in the disjoint layers using simultaneous direct randomized benchmarking sequences, (d) of Fig.~\ref{fig:lf_schematic}. 
\item From each measured decay we obtain a process fidelity $F_i=\frac{1+(d^2-1)\alpha}{d^2}$ where $d$ is the dimension of the decay space ($d=2$ for 1Q, $d=4$ for 2Q) and $\alpha$ is the RB decay rate.  The layer fidelity per disjoint layer is
\begin{equation}
LF_m = \prod_{j} F_{j,m}, \label{eqn:disjoint_lf}
\end{equation}
and the full layer fidelity is
\begin{equation}
LF = \prod_{m}^{M} LF_m. \label{eqn:lf}
\end{equation}
We define a normalized quantity, the error per layered gate,
\begin{equation}
    \mathrm{EPLG} = 1-LF^{1/n_{2q}} \label{eqn:eplg},
\end{equation}
where $n_{2q}$ is the number of 2Q gates in all the layers, e.g., $N-1$ for the minimal set of connected gates
\end{enumerate}

There are a few considerations for the protocol:
\begin{itemize}
    \item  $\{U_{ij}\}$ are typical two-qubit gates such as CNOT, CZ, and iSWAP and variations from those that differ by single qubit gates, e.g., ECR ($e^{-i\frac{\pi}{4}ZX}$).
    \item There is no unique decomposition of disjoint layers, but the error of all qubits must be measured, including idle qubits, i.e., qubits without a two-qubit gate in that disjoint layer. We show some data comparing different disjoint layer decompositions in \S~\ref{sect:add_data}.
    \item A requirement of the protocol is that barriers must be enforced at the layer of two-qubit gates (all gates before the barrier must complete before the circuit can proceed). That is, we apply a set of randomizing single qubit Clifford gates on all qubits, a barrier across all qubits, the layer of  disjoint two-qubit gates, then another barrier, and repeat this $l$ times. At the end, we invert each disjoint set, and measure the ground state population of each akin to simultaneous randomized benchmarking~\cite{gambetta:2012}. Since the sub-layers are disjoint there is no mixing and we get well-defined decay curves of the ground state population versus $l$. The use of barriers keeps each layer consistent with how it would appear in the full (i.e. non-disjoint) layer.
    \item Dynamic decoupling is allowed.
    \item $LF$ makes a Markovianity assumption and is a benchmark insensitive to state preparation and measurement (SPAM) errors; however, the contribution of measurement error can be trivially added by taking the product of the measurement assignment fidelities.
    \item The layer fidelity of a device for $N$ qubits is defined as the maximum layer fidelity measured on the device (practical considerations are discussed in \S~\ref{sect:data}).
\end{itemize}

The goal of the protocol is to measure the fidelity of the full layer defined in the first step of the protocol, for example, (a) of Fig.~\ref{fig:lf_schematic}. Formally, the fidelity of that layer is the trace of the Pauli Transfer Matrix (PTM) between the noisy experimental map and the inverse ideal map,
\begin{equation}
F = \mathrm{Tr}(R_{\mathrm{ideal}}^{-1} R_{\mathrm{exp}})/d^2, \label{eqn:Fexact}
\end{equation}
where $R_{ij} = \mathrm{Tr}(P_i \Lambda[P_j])/d$, $P_i$ are the Pauli matrices, and $\Lambda$ is the process map for the layer. In the limit of no crosstalk, Eqn.~\ref{eqn:disjoint_lf} is exact, but Eqn.~\ref{eqn:lf} is not because the product of traces is not the trace of the product; however, for small errors this is a good approximation (\S~\ref{sect:pf}) and is a lower bound. This is also true for the layer fidelity: after $j$ repetitions of the layer, $LF^{j}$ is approximately the true fidelity until $LF^{j}$ gets small. With crosstalk, Eqn.~\ref{eqn:lf} and Eqn.~\ref{eqn:Fexact} are not identical and for specific crosstalk terms (see \S~\ref{sect:xtalk}) the layer fidelity will be a \emph{lower bound} (the crosstalk error terms are double counted). Because a general treatment of all cases is not possible, we turn to numerics (\S~\ref{sect:sims}) with various noise models. We compare layer fidelity to theory (Eqn.~\ref{eqn:Fexact}) and to the fidelity measured from mirror RB~\cite{proctor:2022}, which is a protocol to measure the layer fidelity by building a circuit of $l/2$ layers to which the reverse circuit is appended, and the polarization of the output is measured versus $l$. 

The advantage of mirror circuit RB is that it does capture all crosstalk terms; however with two distinct disadvantages compared to layer. First, with layer fidelity we obtain more information: a detailed set of error rates for each $\{U_{ij}\}_m$ and $\{q_i\}_m$. Second, the signal to noise of layer fidelity is higher since we are measuring the individual error rates versus the error rate of the entire layer (a weight-$n$ measurement). Any protocol that requires the estimation of high-weight observables, which $LF$ avoids, is unscalable because, with enough qubits, the signal will be unmeasurably small even for short protocol depths. Overall, the numerics support the assertion that layer fidelity is capturing the majority of the crosstalk terms for realistic noise models, and layer fidelity and mirror RB agree well in an experimental test (\S~\ref{sect:add_data}). By fitting all decay terms~\cite{mckay:2020,harper:2020} available to us in our layer fidelity benchmark, we could properly better account for some of these crosstalk terms. However, the added complexity, exposure to measurement errors, and loss of signal-to-noise need more careful consideration.  As an aside, the advantages shown here for layer fidelity over mirror benchmarking should also hold for binary RB \cite{hines:2023} as well. 

As mentioned, one advantage of the layer fidelity protocol is that it gives access to the discrete fidelities of the underlying gates. One utility of this is that we can easily measure the layer fidelity on smaller subsets of the measured set $\{q_i\}$. We can simply calculate the smaller subset by omitting qubits outside the set in the calculation of $LF$, i.e., change the indices of Eqn.~\ref{eqn:disjoint_lf}. A gate may extend outside the new subset, so we assume that the gate fidelity is shared equally between those subsets and calculate the fidelity of that qubit in the layer as $F^{1/2}$. In most geometries the layer fidelity is optimally measured on a long 1D line of qubits (chain of qubits) as this only requires two disjoint layers, i.e., the gates first starting at $Q_0$ (even set) and then at $Q_1$ (odd set) as shown in (b) of Fig.~\ref{fig:lf_schematic}. When defined on a line, measuring subsets is particularly straightforward as it is a sliding window of qubits inside the larger 1D chain.  Although it is not guaranteed that we find the optimal value of layer fidelity using this subspace method, this can be used as a lower bound. While the line is the densest application of gates possible, based on the definition set forth, the gates can be measured over more disjoint layers, so long as idle qubit errors are accounted for; we show data in \S~\ref{sect:add_data} that splitting over more layers is worse due to the increased duration. In certain geometries, such as a star, more disjoint layers will necessarily be required: this problem is equivalent to constructing an edge-coloring of the coupling graph, and Vizing's theorem guarantees that we require no more distinct layers than the degree of the graph plus one.

Layer fidelity can be easily related to other metrics which quantify the error models on a layer, such as $\gamma$~\cite{berg2022probabilistic}, which is defined as 
\begin{equation}
    \gamma = e^{\sum_{k \in K} 2\lambda_k} \label{eqn:gamma1},
\end{equation}
where $\lambda_k > 0$ are the Pauli generator terms in the Lindblad model of the Pauli-twirled noise. While in general all Pauli-twirled error terms exist in Eqn.~\ref{eqn:gamma1}, approximations are made to make the calculation tractable, for example, in Ref.~\cite{berg2022probabilistic} the terms are truncated to all physical connections in the device, and Pauli benchmarking is required to learn them. Even still, this requires many more circuits than are required to measure layer fidelity. The two quantities are easily related for well-behaved noise; for depolarizing noise $\gamma$ is given by, 
\begin{eqnarray}
    \gamma_D & = & \prod_i \left(\frac{16 \times F_i-1}{15}\right)^{-15/8} \nonumber \\
    & = & \prod_i \alpha_i^{-15/8}  \label{eqn:gamma2}
\end{eqnarray} 
which in the limit of $\alpha$ close to 1 is,
\begin{equation}
    \gamma = \frac{1}{LF^2}. \label{eqn:lf_to_gamma}
\end{equation}
We derive this and discuss the bounds in \S~\ref{sect:gamma}, and we show some data comparing $\gamma$ and $LF$ in \S~\ref{sect:add_data}. Note that this $\gamma$ is defined over the disjoint layers used to measure layer fidelity, which are at least depth 2. To estimate $\gamma$ on a $N_{\gamma}$ qubit ($N_{\gamma}$ even), depth 1 layer ($\delta=1$), we can use EPLG (Eqn.~\ref{eqn:eplg}),
\begin{equation}
    \gamma_{\delta=1} = (1-\mathrm{EPLG})^{-N_{\gamma}},
\end{equation}
and 
\begin{eqnarray}
    \bar{\gamma}_{\delta=1} = \gamma_{\delta=1}^{2/N_{\gamma}} = (1-\mathrm{EPLG})^{-2}.
\end{eqnarray}
The accuracy of this estimate will depend on the the layer structures being similar between the layer for $\bar{\gamma}$ and the layer used to measure $LF$/EPLG. If $\gamma_{\delta=1}$ is defined on the same disjoint layer as used for the layer fidelity measurement, then it can also be calculated directly from the disjoint layer fidelity Eqn.~\ref{eqn:disjoint_lf}. 

\section{Data \label{sect:data}}

\begin{figure*}[ht!!]
     \centering
     \begin{tabular}{m{0.95\columnwidth}m{0.95\columnwidth}}
      \includegraphics[width=0.9\columnwidth]{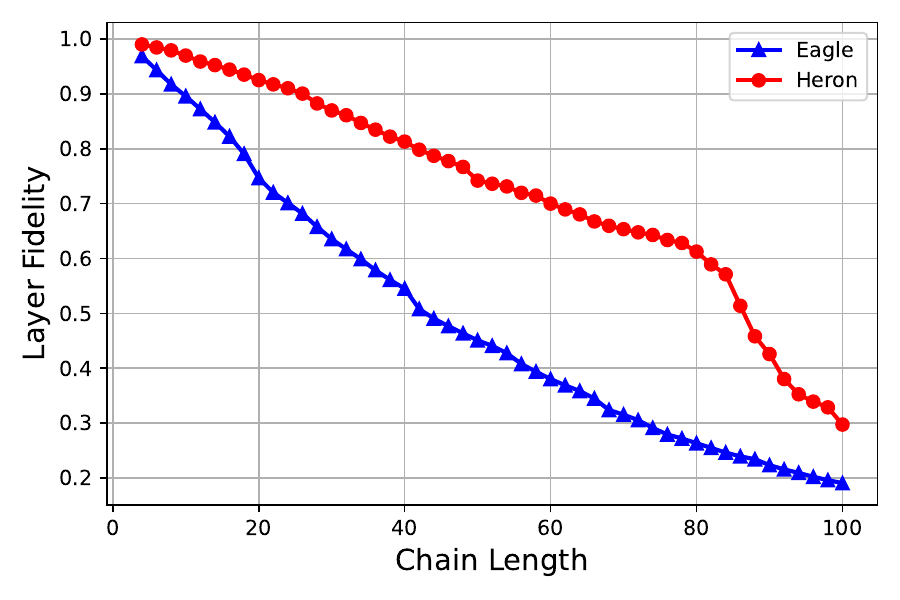} &   \includegraphics[width=0.9\columnwidth]{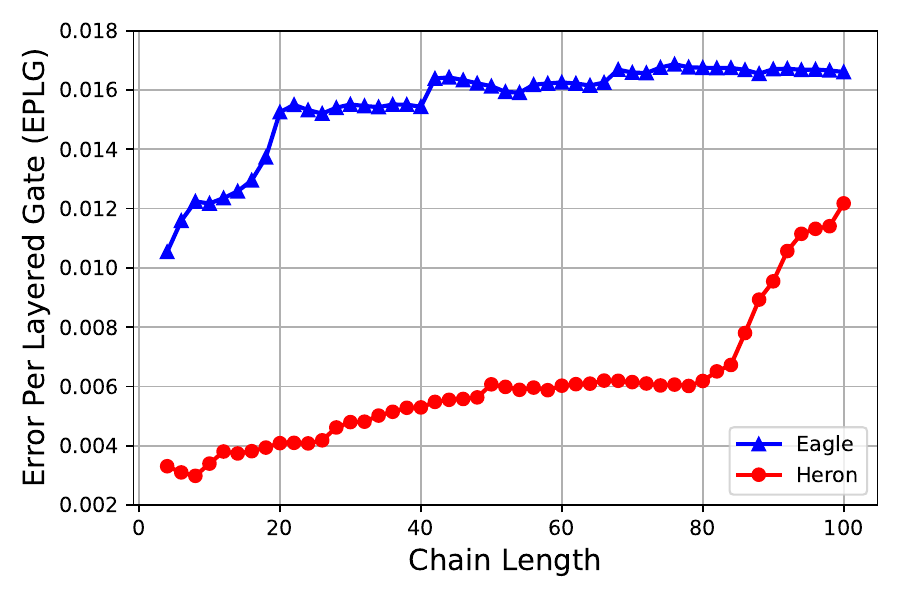} \\
     \includegraphics[width=0.95\columnwidth]{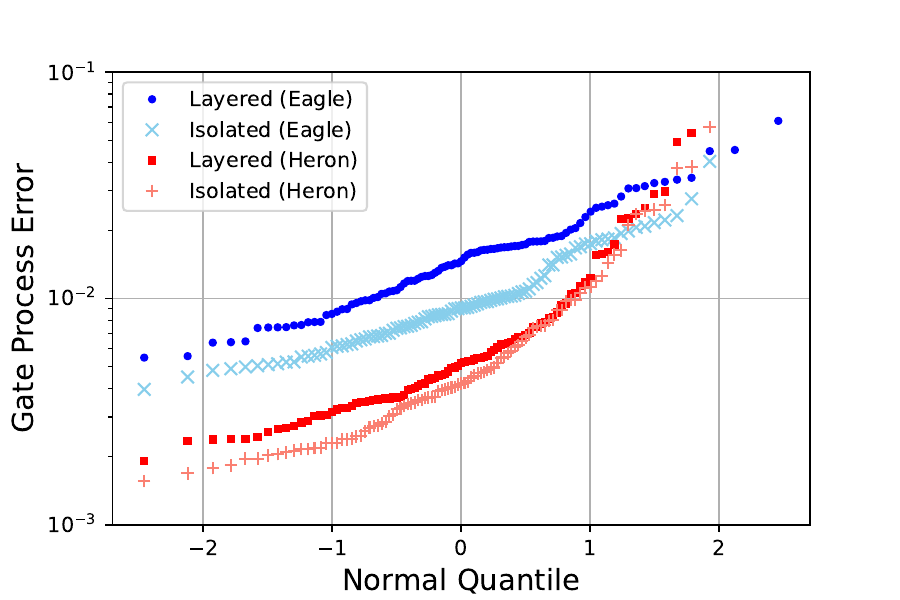} & 
     \raisebox{0.3in}{\includegraphics[width=0.95\columnwidth]{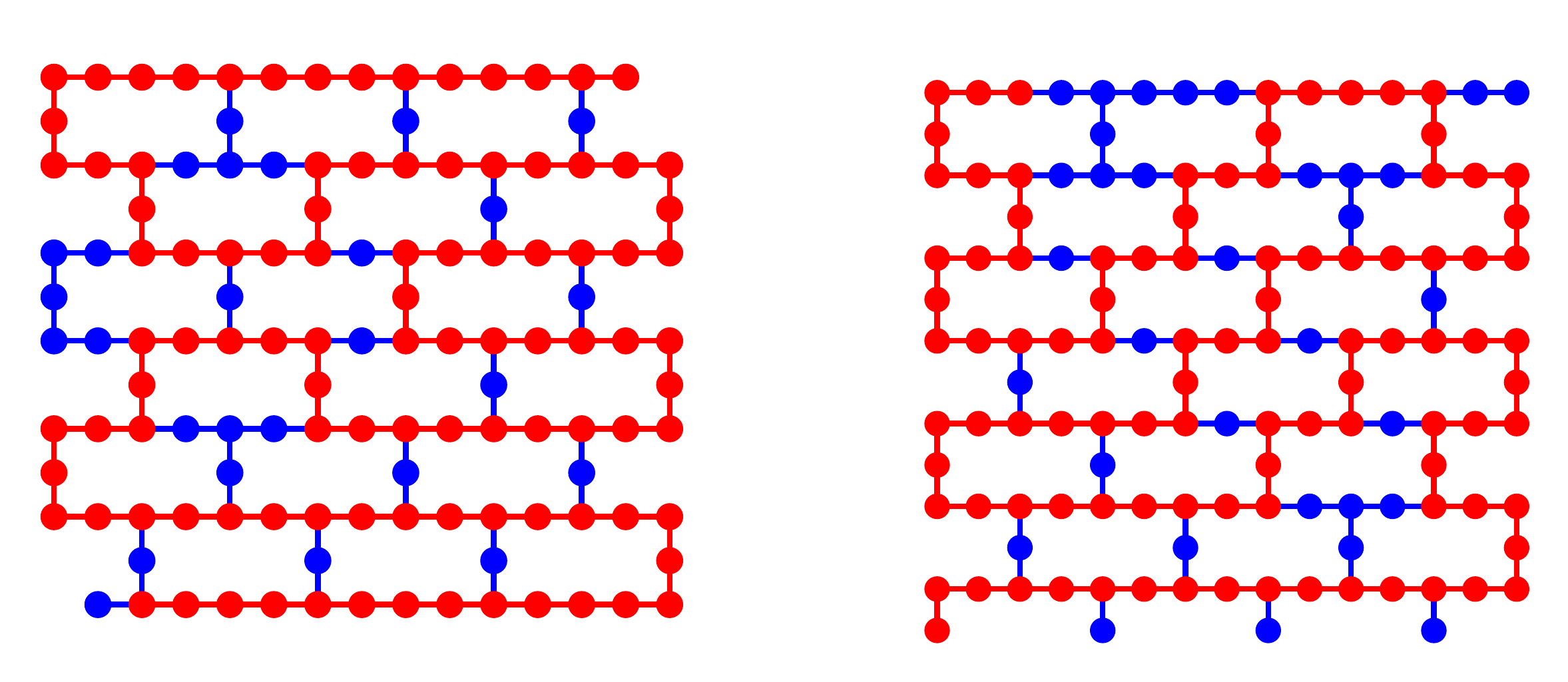} }
     \end{tabular}
         
     \hfill
        \caption{(Top Left) Layer fidelity for the 127 qubit ibm\_sherbrooke ``Eagle'' processor (blue triangles) and the 133 qubit ibm\_montecarlo ``Heron'' processor (red circles) taken using the procedure outlined in the main text for various chain lengths up to 100 qubits. (Top Right) The same data converted to error per layered gate (EPLG). (Bottom Left) Quantile plot of the individual gate errors measured from the best 100 qubit chain from simultaneous direct RB (``layered'') versus the backend reported gate errors (``isolated'') on the same chain. Errors are reported as process error ($\epsilon_p$) as opposed to average gate error ($\epsilon_g$) where $\epsilon_p=\frac{d+1}{d} \epsilon_g$. Both devices have among the lowest gate error measured on a superconducting device, noting the minimum isolated gate error (process error) on ibm\_sherbrooke (Eagle) of $3.2(4.0) \times 10^{-3}$ and on ibm\_montecarlo (Heron) of $1.2(1.6) \times 10^{-3}$  (Bottom Right) The 100 qubit chain (red) overlaid on the ibm\_sherbrooke (left) and ibm\_montecarlo (right) device layout schematics. }
        \label{fig:lf_data}
\end{figure*}

As mentioned, the layer fidelity for a subset of $N$ qubits on a device ($LF_{N}$) is defined as the maximum layer fidelity over all K ($N$-qubit)  subsets. Practically it will be impossible to measure all sets on a large device, for example the number of length 100 chains on a 127Q heavy hex device such as ibm\_sherbrooke is 313,980. Therefore, in practice we need heuristic methods to measure the optimal layer fidelity. Initial estimates of the layer fidelity can be made with the isolated two-qubit fidelities~\cite{nation2022suppressing} and from there candidate sets can be measured. One of the bigger considerations here is that the layer fidelity imposes a fixed length on all gates of the disjoint layer equal to the longest gate (see the simulations in the appendix \S~\ref{sect:sims}), and so the estimates from isolated RB fidelities must take that into consideration. Typically, this is done by omitting edges of the graph with gates that are much longer than the average. We use a heuristic protocol given by the following procedure:

\begin{enumerate}
    \item Assuming a list of gate errors measured from isolated RB is available, calculate the layer fidelity for each $N_{max}$ qubit linear string, where $N_{max}$ is selected to be at least the length of the longest desired chain. In this step long gates may be omitted from the graph as they are known to make the layer fidelity much worse. Find the set with the highest predicted $LF$ (set 1), then find the set with the least overlap with set 1 and the highest predicted $LF$ of that subset (set 2). Repeat this again to find a third set.
    \item Measure the errors from simultaneous direct RB (described in Fig.~\ref{fig:lf_schematic}) in those 3 sets (at least 6 disjoint layers for 1D chains) and calculate the $LF$ (Eqn.~\ref{eqn:lf}) from the measured data for each $N<N_{max}$ by looking at subchains within the sets.
    \item For each value of $N$ take the largest $LF$ from all the subchains measured. For example, if $N_{max}=100$ and $N=50$ then there are 150 possible sub-chains.
    \item Plot $LF$ vs $N$, convert to error per layered gate (EPLG) as $\mathrm{EPLG}=1-LF^{1/n_{2q}}$ where $n_{2q}$ is the number of 2Q gates.
    \item Since this covers a heuristic number of chains, more chains at different lengths can be measured ``ad-hoc'' and if the $LF$ of those chains is larger, they will supplant the previously measured values.
\end{enumerate}

We show typical data taken on a 127 qubit ``Eagle'' processor ibm\_sherbrooke (native two-qubit CX gate using cross-resonance) and 133 qubit ``Heron'' processor ibm\_montecarlo (native two-qubit CZ gate using tunable-coupler actuation) in Fig.~\ref{fig:lf_data}. To measure the fidelities we perform the simultaneous direct RB sequences described previously with 300 shots per circuit, 6 randomizations and $l=$~[1, 10, 20, 30, 40, 60, 80, 100, 125, 150, 200, 400] (Eagle) and $l=$~[1, 10, 20, 30, 40, 60, 80, 100, 125, 150, 200, 400, 750] (Heron). For ibm\_sherbrooke most gate lengths are 533~ns, but as described in the heuristics for picking the gate sets, three edges with gate lengths $>$700~ns were removed from consideration. For ibm\_montecarlo the gate lengths were about an equal mix of 84~ns and 104~ns and none were removed from consideration due to gate length. All circuits were generated in Qiskit and run through the IBM Quantum cloud interface. In the top plot we show the layer fidelity and the error per layered gate as a function of chain length. The chain of qubits on the ibm\_sherbrooke and ibm\_montecarlo devices for the $N=100$ data is shown in red on the bottom right plot. 

One of the advantages of this layer fidelity measurement is the access to individual gate errors which can be used for further analysis or fidelity estimates. In particular, we can perform a comparison between isolated RB and the errors from layer fidelity RB which can be a proxy for crosstalk errors. We note here that the isolated RB data is, in fact, a variant of simultaneous RB where there is a distance of at least one idle qubit between all two qubit gate pairs and there are no barriers. As we show in simulations in the appendix \S~\ref{sect:sims} isolated RB trivially eliminates some crosstalk, such as always on ZZ between pairs of qubits for fixed coupling architectures. The data bears this out as the middle plot of Fig.~\ref{fig:lf_data} shows a distinct increase in the error per gate on the ``Eagle'' processor when run in layers versus isolated RB mode. Conversely, these errors are greatly alleviated in the ``Heron'' processor since the coupling between neighboring qubits can be turned off when not required for two-qubit gate operation.

\section{Conclusions}

In this manuscript we discussed a benchmark for quantum processors at scale - the layer fidelity. The layer fidelity follows naturally from standard randomized benchmarking procedures, is crosstalk aware, fast to measure over a large number of qubits, has high signal to noise and gives fine-grained information. We demonstrated the key components of the layer fidelity metric with measurements on the 127 qubit Eagle processor ibm\_sherbrooke and 133 qubit Heron processor ibm\_montecarlo. Using simulation and data we showed (\S~\ref{sect:add_data} and \S~\ref{sect:sims}) that there is good agreement with mirror randomized benchmarking - a complementary technique for measuring layers - over a number of error models. The layer fidelity links easily with other methods of characterizing layers, such as Pauli learning for $\gamma$. We leave a few open questions outside the scope of this manuscript, such as whether more advanced data fitting can improve agreement with the exact layer fidelity, the predictive power of $LF$ with differently structured circuits, the limits of twirling in $LF$, and how to extend to layered circuits with mid-circuit measurements. Finally, we note that, like all benchmarks layer fidelity should be considered as one piece of information towards full device characterization. 

\begin{acknowledgments}
The authors would like to thank Andrew Wack for helpful discussions, Kevin Krsulich for Qiskit help, Samantha Barron for help with the Pauli learning, and James Wootton and Blake Johnson for manuscript comments. Research was sponsored by the Army Research Office and was accomplished under Grant Number W911NF-21-1-0002. The views and conclusions contained in this document are those of the authors and should not be interpreted as representing the official policies, either expressed or implied, of the Army Research Office or the U.S. Government. The U.S. Government is authorized to reproduce and distribute reprints for Government purposes notwithstanding any copyright notation herein.
\end{acknowledgments}

\bibliography{main}

\begin{appendix}

\section{Additional Data \label{sect:add_data}}

Here we provide some additional support for layer fidelity by comparing to the mirror RB protocol~\cite{proctor:2022} which embeds the full layer directly in a mirror circuit. Our specific mirror circuit is comprised of a random 1Q Clifford layer, then the first disjoint layer of two-qubit gates, a second random 1Q Clifford layer, then the second disjoing layer of two-qubit gates. In this way the number of 1Q gates is the same between mirror RB and layer RB when constructing the full layer fidelity. We consider two versions of mirror, one which is a direct mirror (just the forward and reverse circuit) and the second version, more faithful to Ref.~\cite{proctor:2022}, includes a random Pauli layer between the forward and reverse circuits. We measure the polarization $S$ as defined in Ref.~\cite{proctor:2022}.  Our data comparison is on 20 qubits of the ibm\_peekskill device, which is a 27 qubit fixed-coupling ``Falcon'' processor and the path is shown in Fig.~\ref{fig:mirror_data}. We perform 10 randomizations and measure 2000 shots. We compare to layer fidelity data from 6 randomization and 300 shots. We can see that the agreement between mirror and layer is quite good and that there is some discrepancy between mirror with and without the Pauli layer. Since our error model in the device is fixed, we investigate the comparison of layer and mirror further with simulations in \S~\ref{sect:sims}. 

Another aspect of the layer protocol is that the choice of disjoint sets is not unique, as shown in Fig.~\ref{fig:lf_schematic} for 2 versus 4 layers of the chain. Therefore, we take data in the same set of qubits as above splitting the layer fidelity into 2, 4, 6 and 10 disjoint layers. Because of the increased total duration, more disjoint layers leads to lower fidelity. However, this statement will be architecture dependent; there are certain types of crosstalk terms that occur during simultaneous gates that are large enough to offset the longer duration of the circuit, e.g., this was probed on IBM devices in Ref.~\cite{murali2019noise}. In this case, splitting into more disjoint layers is a sensible approach. Another scenario is that the architecture does not allow more than a certain number of simultaneous gates at a time.

\begin{figure}[ht!]
     \centering
         \includegraphics[width=\columnwidth]{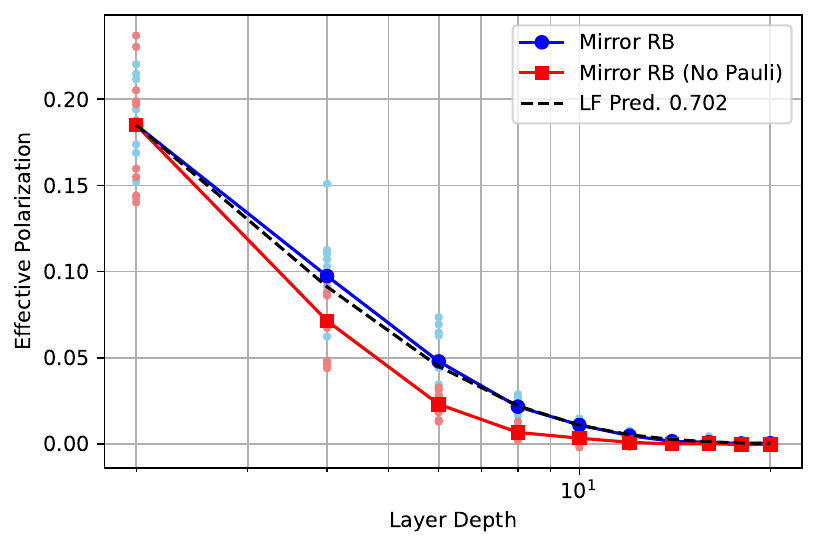}
         \includegraphics[width=\columnwidth]{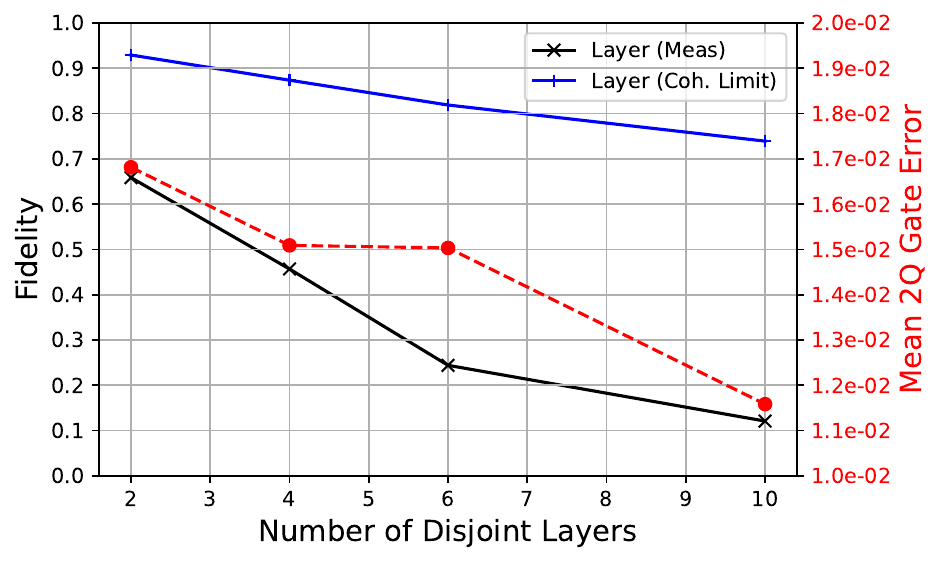}
         \includegraphics[width=0.75\columnwidth]{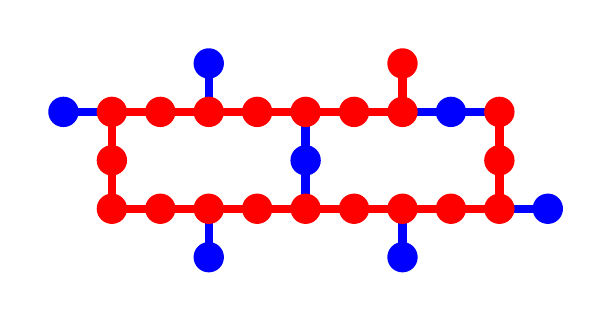}
     \hfill
        \caption{(Top) Comparing mirror RB data versus layer depth $l$ to the predicted decay from layer fidelity measured on the same set of qubits; for this data $LF=0.702$ and so the dashed line is $0.702^l$. The different mirror RB curves either include (blue) or do not include (red) a random Pauli layer. (Middle) The layer fidelity versus the number of disjoint layers used in the protocol (black). The fidelity decreases as the number of layers increases because the total duration is longer. We estimate this effect by just considering the fidelity decrease due to decoherence (blue). More layers does decrease the mean gate error (dashed, red) due to lower crosstalk, but overall this is not enough to improve because of the increased length. (Bottom) Qubits used on ibm\_peekskill are [23, 24, 25, 22, 19, 16, 14, 11, 8, 5, 3, 2, 1, 4, 7, 10, 12, 15, 18, 17].}
        \label{fig:mirror_data}
\end{figure}

In the main text we relate the layer fidelity to a quantity relevant for error mitigation, $\gamma$, defined in Eqn.~\ref{eqn:gamma1}. The relation is given in Eqn.~\ref{eqn:lf_to_gamma}, and although supported in theory by well behaved noise models (see \S~\ref{sect:gamma}), here we do an experimental comparison on a 16 qubit section of ibm\_peekskill. The layer fidelity data is taken according to the procedure outlined in the main text and the direct $\gamma$ data is taken according to the procedure in Ref.~\cite{berg2022probabilistic}. For the LF data we take 178 circuits (13 depth points $\times$ 6 randomizations $\times$ 2 disjoint layers) and for the $\gamma$ data we take 14,000 circuits (14 depth points $\times$ 1000 basis rotations $\times$ randomizations). This ratio of circuits demonstrates why layer fidelity is a quick method for estimating gamma. The data is shown in Fig.~\ref{fig:gamma_data} and the agreement is reasonable; a more comprehensive study including error bars and minimizing time variations of the device properties is left for a future study. 

\begin{figure}[ht!]
     \centering
         \includegraphics[width=\columnwidth]{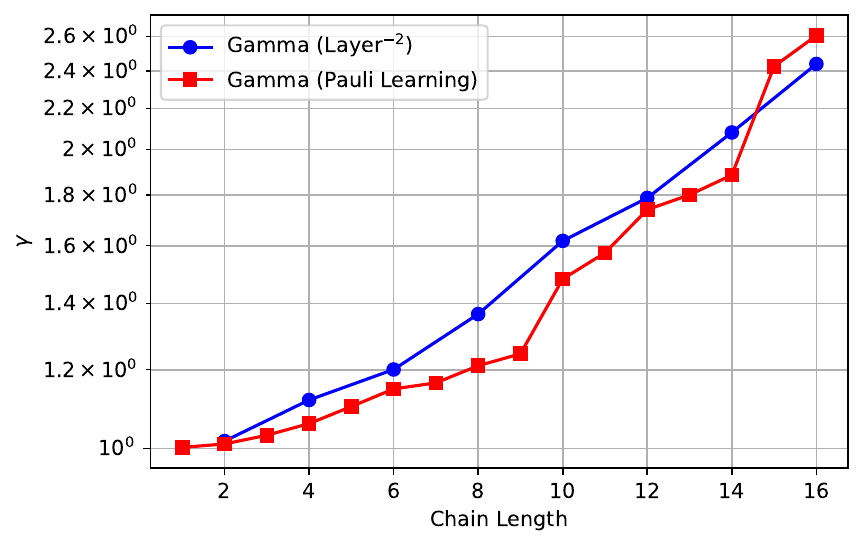}
     \hfill
        \caption{Comparing $\gamma$ measured from layer fidelity (blue, circles) and Eqn.~\ref{eqn:lf_to_gamma} to $\gamma$ measured using Pauli-learning~\cite{berg2022probabilistic} (red, squares). Measured on ibm\_peekskill for the connected set of qubits [19, 22, 25, 24, 23, 21, 18, 15, 12, 13, 14, 11, 8, 5, 3, 2] (even and odd disjoint layers).}
        \label{fig:gamma_data}
\end{figure}

\section{Simulations \label{sect:sims}}

Here we compare simulations between isolated RB, simultaneous RB, layer fidelity RB, and mirror RB with a variety of error models. The circuits are generated as gates (in the decomposition of $X90$, $X0$ [idle], $Rz$ [arbitrary Z rotations] and $CX$ gates) and then converted to a schedule based on the single-qubit gate being the smallest unit of time; the two-qubit gates are converted into fractional time steps of either 5 or 8 single-qubit gate times. Because the $Rz$ gates are zero time, they are considered their own gate slices, and so a finite time step can consist of the 3 possible gates on each qubit and so for four qubits there are roughly 81 unique four qubit unitaries to construct. We may add coherent error terms to each unitary, e.g., an overrotation or a $ZZ$ crosstalk. We then perform a density matrix simulation, where at each time step the unitary is applied followed by a discrete $T_1$/$T_2$ map on each qubit. Breaking the unitary and incoherent evolution into time steps is an implementation of Trotterized simulation of their concurrent evolution. We compare the error extracted from these sequences to the theoretical errors by adding the coherent~\cite{emerson:2005,pedersen:2007} and incoherent errors,
\begin{eqnarray}
\epsilon_{U} & = & 1-\frac{\textrm{Tr}\left[U_{ideal}^* U)\right]^2}{d^2} \label{eqn:coh_err} \\
\epsilon_{\Lambda} & = & 1-\prod_{i}\left(\frac{1}{4}+\frac{1}{2}e^{-t_g/T_{2,i}}+ \right. \nonumber \\
& & \left. \frac{1}{4}e^{-t_g/T_{1,i}}\right).
\end{eqnarray}
where $t_g$ is the gate length and both these errors are process errors. 

\begin{figure}[ht!]
     \centering
         \includegraphics[width=\columnwidth]{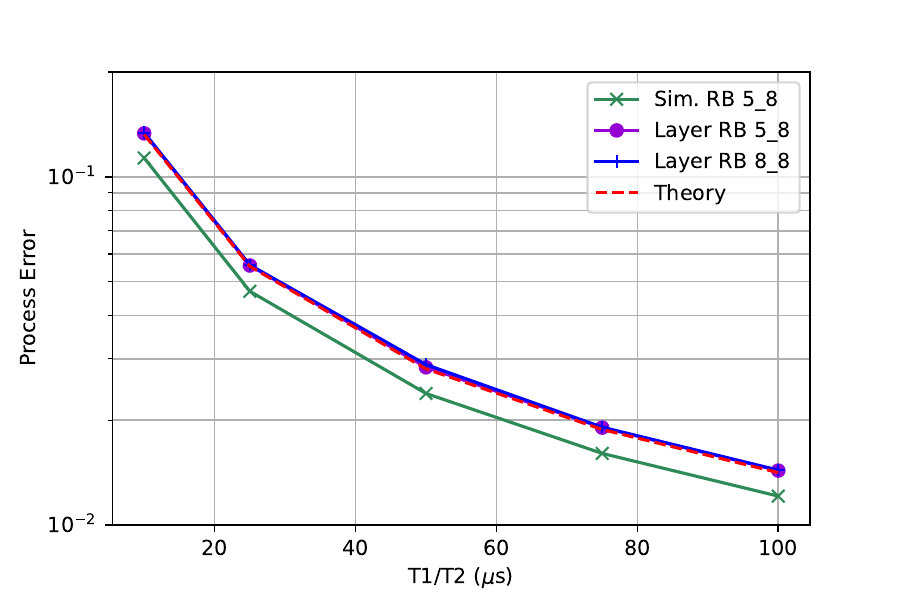}
         \includegraphics[width=\columnwidth]{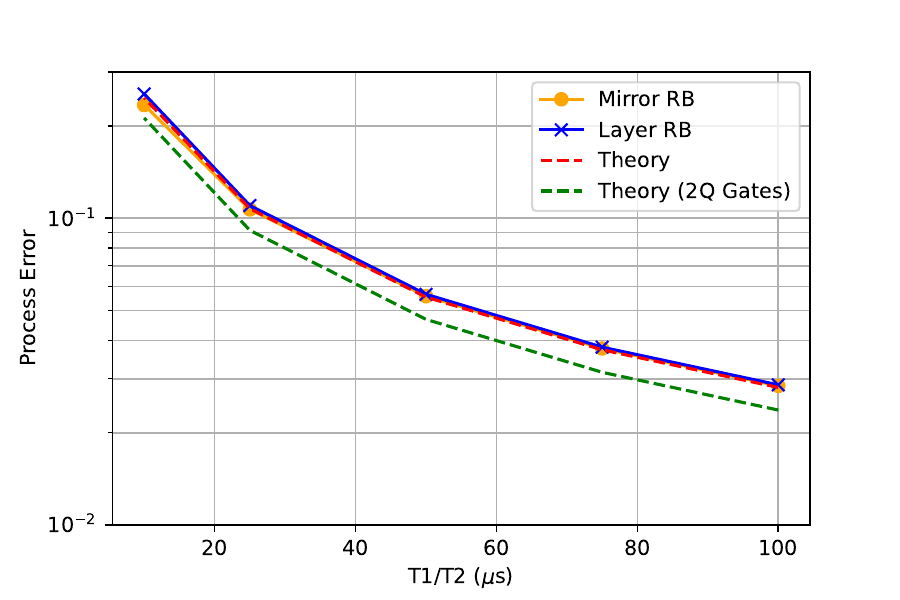}
     \hfill
        \caption{(Top) Simulation of the even layer with incoherent errors ($T_1=T_2$) and the gate unit length of 50~ns. As described in the main text, when the gate lengths are different simultaneous RB trivially gives the wrong answer. (Bottom) Similar simulation for the two layers with incoherent errors ($T_1=T_2$) comparing mirror to layer and the agreement is exact. There are two theory curves in the bottom plot; in the red curve the single qubit gates and the two qubit gate layers are included in calculating the total incoherent error (there are on average 1.5 single qubit gate layers per two-qubit gate layer). The green theory curve is the error if we just consider the two-qubit layer. For agreement with theory, the single qubit gates in the layer must be considered. There are 10 random sequences in each simulation.}
        \label{fig:lf_sim_t1t2}
\end{figure}

In the first set of simulations we only consider incoherent errors and perform the simulation on the even layer of the 4Q set as shown in the top of Fig.~\ref{fig:lf_sim_t1t2}. We consider two different scenarios, one where the two gates in the layer ($CX_{01}$ and $CX_{23}$) are different lengths (5 time units for $CX_{01}$ and 8 time units for $CX_{23}$) and another scenario where both gates are 8 time units. Trivially, simultaneous RB gives the wrong answer for the error of the layer because there are no enforced barriers between the different two-qubit gates. The different layer fidelities are the same because of the barrier. This illustrates how the layer fidelity enforces the layer to be as long as the longest gate for all qubits. The theory agrees well, once we include the 1.5 single qubit gates per layer (so the layer is considered 8+1.5 units in length). The length unit is 50~ns. We then continue the simulation for both layers (with all three gates having length 8) and compare to mirror RB (bottom of Fig.~\ref{fig:lf_sim_t1t2}). For comparison here the mirror layer has a set of random 1Q Cliffords before each layer of two-qubit gates so that the total gate counts are the same in layer and mirror. The agreement between the two is near exact. 

\begin{figure}[ht!]
     \centering
         \includegraphics[width=\columnwidth]{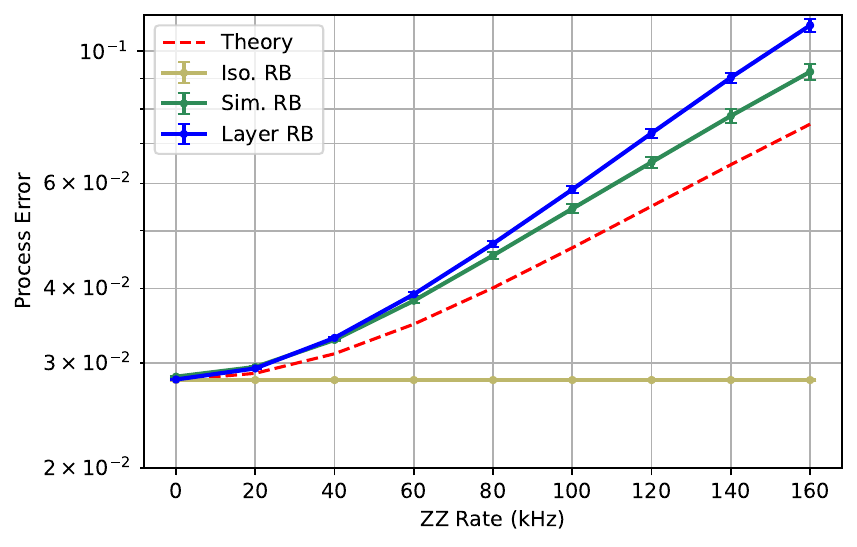}
         \includegraphics[width=\columnwidth]{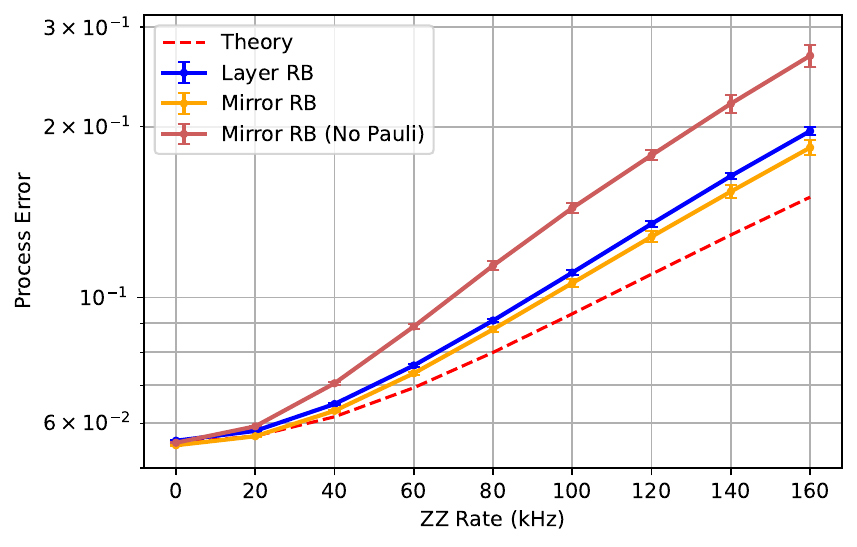}
         \includegraphics[width=\columnwidth]{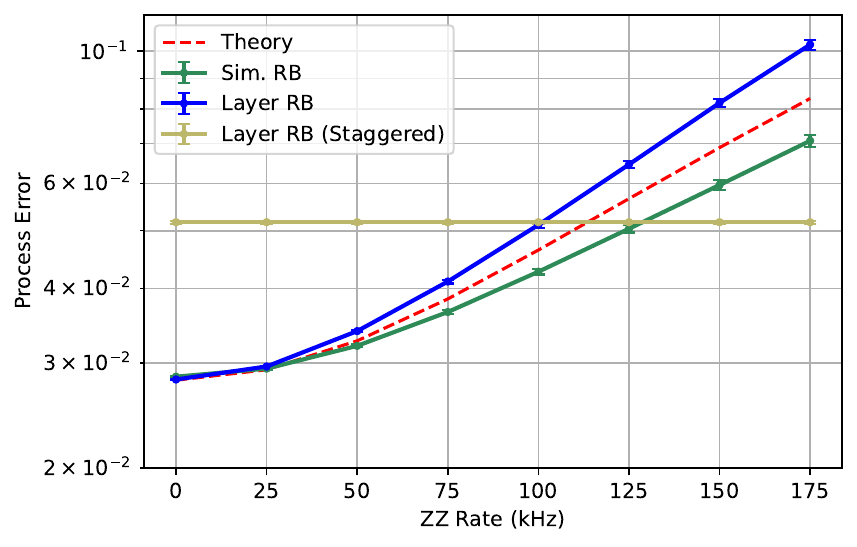}
     \hfill
        \caption{(Top) Simulation of the even layer vs ZZ rate with the same length gate (8 units) comparing isolated RB, simultaneous RB and layer RB. (Middle) Simulation of the full layer vs ZZ rate comparing layer RB, mirror RB and mirror RB without a Pauli layer between mirrors. (Bottom) Simulation of the even layer vs ZZ, where the ZZ is only applied when there are simultaneous 2Q gates. If we stagger the gates then the crosstalk term disappears, but the overall baseline error is higher. There are 30 random sequences in each simulation.}
        \label{fig:lf_sim+zz}
\end{figure}

Next we investigate the more interesting case of coherent crosstalk error. We take $T1=T2=50\mu s$ (same unit time length as before, 50~ns) and vary the $ZZ$ interaction rate, $e^{-i 2\pi \xi_{ZZ} |11\rangle \langle 11|}$, between qubits 0 and 3 ($ZZ_{03}$) and qubits 1 and 2 ($ZZ_{12}$). This error is out of the disjoint subspace.  We consider two versions of this $ZZ$ crosstalk; one version where the $ZZ$ is ``always-on'', and another where it only occurs during simultaneous two-qubit gate operation. All the gate lengths are the same (8 units). 

First, we look at just layer 1 (top Fig.~\ref{fig:lf_sim+zz}) with always-on $ZZ$ and compare isolated RB, simultaneous RB and layer RB. Trivially the isolated RB is not affected by the ZZ interaction, demonstrating that it's a poor method for assessing crosstalk. Simultaneous RB and Layer RB are reasonably similar, with the caveat from Fig.~\ref{fig:lf_sim_t1t2} that if the gate lengths are different simultaneous RB will not reflect the layer properly. 

Next we consider the full layer with always-on $ZZ$ and compare layer RB to mirror RB. We look at two flavors of mirror RB, the first is to exactly mirror the circuit (``no pauli'') and the second is to more faithfully execute the mirror circuit with a random Pauli layer between the  original circuit and its mirror. For this crosstalk the ``no Pauli'' mirror reports a much higher error whereas the layer and mirror (with Pauli) are in very close agreement. 

Finally, we look at layer 1, with $ZZ$ that is only activated by simultaneous 2Q gates, noting that such a crosstalk could be activated by the physics described in Ref.~\cite{wei:2022,mitchell:2021}. Here there is a greater divergence between the layer fidelity and simultaneous RB results because simultaneous RB does not enforce strictly running the 2Q gates at the same time. Furthermore, we see that if we run a layer where the 2Q gates are staggered, this crosstalk term trivially vanishes, although the layer has higher baseline error since it's longer. This elucidates why sometimes it can be beneficially to run non-simultaneous gates for crosstalk as seen in Ref.~\cite{murali2019noise}. This particular version of crosstalk is not necessarily representative, but serves as an example for a family of similar crosstalk terms that activate with simultaneous 2Q gates. We note that there is some ambiguity in calculating the theory curves for these plots, we use Eqn.~\ref{eqn:coh_err}, but since the errors are in the single qubit layer as well we approximate the errors by  assuming the single layers and two qubit layers add.

\begin{figure}[ht!]
     \centering
         \includegraphics[width=\columnwidth]{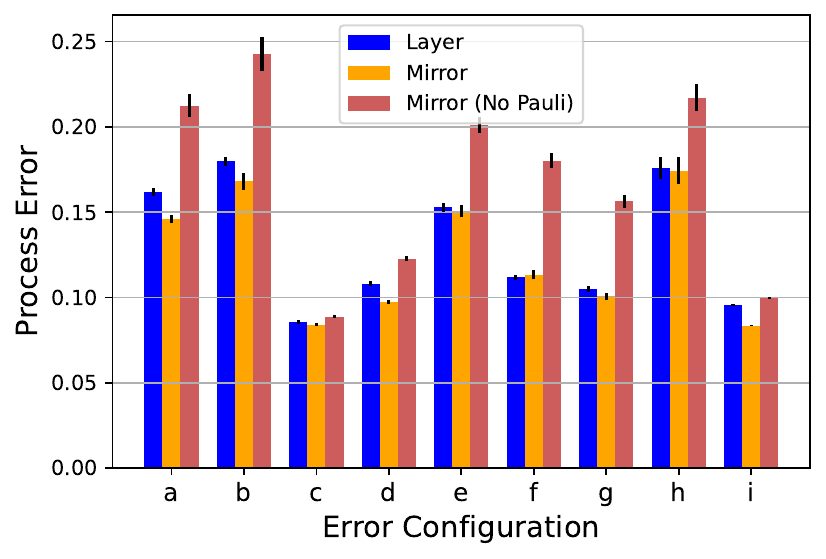}
     \hfill
        \caption{Comparing layer fidelity, mirror RB and mirror RB without a Pauli layer for measuring the process error with several different coherent error scenarios. As in the other simulations $T_1=T_2=50~\mu$s and the unit time is 50~ns. (a) Always on 150~kHz $ZZ$ rate between $0\_1$ and $2\_3$. (b) Always on 150~kHz $ZZ$ rate between $0\_3$ and $1\_2$ (same as Fig.~\ref{fig:lf_sim+zz}). (c) Simultaneous only 150~kHz $ZZ$ rate between $0\_1$ and $2\_3$. (d) Simultaneous only 150~kHz $ZZ$ rate between $0\_3$ and $1\_2$. (e) Always on 100~kHz $ZZ$ rate between $0\_1$, $1\_2$ and $2\_3$ (all the connected qubits). (f) $Z$ error applied after every time slice of 0.02, $e^{-i0.02Z/2}$. (g) 10\% over rotation on all two-qubit gates. (h) 10\% over rotation on all two-qubit gates and a 10\% under rotation on all 1Q gates. (i) Drive crosstalk of 10\% (IY and ZY) from qubit 1 to qubit 2 when applying the CX$_{01}$ gate, from qubit 2 to qubit 1 when applying the CX$_{23}$ gate, and from qubit 1 to qubit 0 when applying the CX$_{12}$ gate. }
        \label{fig:lf_sim_bar}
\end{figure}

Ultimately the comparison we desire is between layer fidelity and mirror fidelity. In the plots we see the two are fairly close, but the space of possible unitary errors is very large.  Therefore, we take a scattershot look at a variety of different error terms and compare between the two methods, as summarized in Fig.~\ref{fig:lf_sim_bar}. The general trend appears is that layer and mirror measure very similar errors, layer fidelity tends to measure slightly higher error than mirror (consistent with the discussion in the next section), whereas without the Pauli layer mirror always measures a larger error. 

\section{Crosstalk and Layer Fidelity \label{sect:xtalk}}

Here we consider the layer fidelity protocol with a single Pauli weight-2 coherent crosstalk term. There are two subspaces $k$ and $j$ that both have $n_k$ and $n_j$ qubits and there is a crosstalk term of the form,
\begin{equation}
U = e^{-i\alpha P_{x}} \approx \mathcal{I} - i\alpha P_{x} - \frac{\alpha^2}{2} \mathcal{I}
\end{equation}
where $P_x$ is a weight-2 Pauli spanning $k$ and $j$ and $\alpha$ is small so we take the small $\alpha$ expansion. The true fidelity of the layer (idles and this crosstalk term) is 
\begin{eqnarray}
F_U & = & \frac{Tr(U)^2}{4^{n_k+n_j}} \\
& \approx & 1-\alpha^2 \label{eqn:trueF}
\end{eqnarray}

Now, what if we do simultaneous RB and are able to twirl $k$ and $j$ (mythically here without additional problems), from the simultaneous paper we know that the decay parameter is $\mathrm{Tr}[\Pi_k R_U]/\mathrm{Tr}[\Pi_k]$ where $\Pi_k$ are the Pauli's just in $k$ (e.g. if $n_k=2,n_j=2$, this would be the 15 $XIII,YIII,\ldots,IXII,\ldots, XXII$). Calculating the PTM terms (remember we only need the on-diagonal terms) and leaving off the $1/2^{n_k+n_j}$ in front of the trace,
\begin{eqnarray}
(R_{U})_{i,i} & = & \mathrm{Tr}[P_i U^{\dagger} P_i U] \\
(R_{U})_{i,i} & = & \mathrm{Tr}[P_i (\mathcal{I}(1-\frac{\alpha^2}{2}) + i\alpha P_{x}) P_i \nonumber \\
& & (\mathcal{I}(1-\frac{\alpha^2}{2}) - i\alpha P_{x}) ] \\
& = & \mathrm{Tr}[\mathcal{I}(1-\frac{\alpha^2}{2})^2 + i\alpha (1-\frac{\alpha^2}{2}) (P_i P_{x} P_i - P_{x})  + \nonumber \\
& & \alpha^2 P_i P_{x} P_i P_{x}) ] \\
& \approx & \mathrm{Tr}[\mathcal{I}(1-\alpha^2) +  \alpha^2 P_i P_{x} P_i P_{x}) ]
\end{eqnarray}
because the middle terms have trace zero. So if $[P_i,P_x]=0$ then the above is 1, and if they don't commute then the above is $1-2\alpha^2$. In the space for the decay parameter of $k$ then, there are $4^{n_k-1/2}$ elements with $1-2\alpha^2$ (because there are 2 Pauli's in $k$ that don't commute with $P_{X}$ which is weight 2 but only has 1 weight in $k$) and the rest are 1. The fidelity in space $k$ is then
\begin{eqnarray}
F_{U,k} & = & 1 - 2\frac{4^{n_k-1/2}}{4^{n_k}}\alpha^2 \\
& = & 1-\alpha^2
\end{eqnarray}
So the estimate is independent of $n_k$ and therefore $F_{U,k}=F_{U,j}$. And then since we multiply the two fidelities together to estimate $F_U$,
\begin{eqnarray}
\tilde{F}_U & = & (1-\alpha^2)^2 \\
& \approx & 1-2\alpha^2
\end{eqnarray}
which is a lower fidelity than the true fidelity Eqn.~\ref{eqn:trueF}.

This analysis also holds by the same arguments for a Pauli stochastic error channel of the form
\begin{align}
    \mathcal{E}(\rho) = (1-\alpha^2)\rho + \alpha^2P_x\rho P_x,
\end{align}
with $P_x$ defined as before. For this error channel we again find that $F_{\mathcal{E}} = 1-\alpha^2$, and $F_{\mathcal{E},k} = F_{\mathcal{E},j} = 1-\alpha^2$, such that the layer fidelity is a lower bound for the true process fidelity $F_{\mathcal{E}}$.

\section{Combining Process Fidelities \label{sect:pf}}

Here we summarize some properties of the process fidelity which have been shown in other sources for reference. The process fidelity is defined as the trace of the superoperators (see, e.g. Ref~\cite{greenbaum:2015} for a summary), in Pauli form, 
\begin{eqnarray}
R_{ij} & = & \frac{\mathrm{Tr}\left[P_i \Lambda[P_j]\right]}{d} \\
F_{p} & = & \frac{\mathrm{Tr}\left[R_{\mathrm{ideal}}^{-1} R\right]}{d^2}
\end{eqnarray}
which is related to the average gate fidelity~\cite{nielsen:2002}
\begin{equation}
    F_{g} = \frac{dF_{p}+1}{d+1}
\end{equation}
where $\epsilon_g = 1-F_{g}$ is the average gate error and is often quoted from the gate error from randomized benchmarking. If there are two disjoint subspaces that have process fidelities $F_{p,0}$ and $F_{p,1}$, then the fidelity of the combined system is $F_{p,0} F_{p,1}$, which is the property we have used to build up each disjoint layer fidelity. It is, however, not true that process fidelities multiply across layers (for simplicity assuming the ideal is the identity),
\begin{eqnarray}
F_{p,q} & = & \frac{\mathrm{Tr}\left[R_p R_q\right]}{d^2} \\
& \ne & \frac{\mathrm{Tr}\left[R_p\right]}{d^2} \times \frac{\mathrm{Tr}\left[R_q\right]}{d^2}
\end{eqnarray}
However this is approximately true for small errors of diagonal maps, which can be shown by a simple expansion. Practically this means that the the fidelity of the layer repeated to multiple depths is fairly well approximated until the fidelity drops below a percent. 

\section{Relating $\gamma$ to $LF$ \label{sect:gamma}}

In this section we relate the $LF$ to $\gamma$ as defined in Ref.~\cite{berg2022probabilistic}. This is a useful metric for error mitigation since it indicates the number of circuit randomizations required to perform probabilistic error mitigation. $\gamma$ is defined for Pauli diagonal noise model, and although we define $LF$ for a depolarizing model, we will do a general comparison here for a Pauli diagonal noise model. As a reminder the definition of $\gamma$ in terms of the PTM elements defined in the above section is,
\begin{eqnarray}
    \gamma & = & e^{2\sum_k \lambda_k} \\
    R_i & = & e^{-2\sum_{\langle k \rangle_i} \lambda_k} \\
    F_p & = & \frac{1 + \sum_{i}^{4^n-1} R_i}{4^n}
\end{eqnarray}
where $\langle k \rangle_i$ is the sum over $k$ where $[P_i,P_k]\ne 0$ and $\lambda_k$ are generators of a Lindblad equation that are small for small errors. In that limit, it's straightforward to expand the exponentials,
\begin{eqnarray}
    R_i & \approx & 1-2\sum_{\langle k \rangle_i} \lambda_k \\ 
    F_p & \approx & \frac{1 + \sum_{i}^{4^n-1} \left(1-2\sum_{\langle k \rangle_i} \lambda_k\right)}{4^n} \\
    & = & 1 - \sum_{k} \lambda_k \\
    & \approx & e^{-\sum_{k} \lambda_k} \\
    & = & \gamma^{-1/2}
\end{eqnarray}
using the fact that $\sum_{i} \sum_{\langle k \rangle_i}=\sum_{k} \sum_{\langle i \rangle_k}$ and $\sum_{\langle i \rangle_k} = 2^{n}$. \\

Next we explore the correspondence of $\gamma$ to more commonly used gate metrics such as the diamond norm or the average gate fidelity with more rigor and provide bounds.  For a Pauli channel, both the average gate fidelity and the diamond norm have simple deviations from the spectral properties of the process matrix $\Lambda$, where $\Lambda = e^{\cal L}$~\cite{magesan:2012}.  In particular, for a Pauli Channel we can write the average gate error and diamond norm as
\begin{align}
\varepsilon(\Lambda) &= \frac{1-\frac{{\rm Tr}(\Lambda)}{d^2}}{1+\frac{1}{d}},\\
\| \Lambda \|_\diamond &= 2 \left(1+\frac{1}{d}\right)\varepsilon(\Lambda) =2 \left(1-\frac{{\rm Tr}(\Lambda)}{d^2} \right) . \\  
\end{align}
For a Pauli channel, we can derive either metric from the arithmetic average of the eigenvalues of $\Lambda$.

Note that the form in Eqn.~\ref{eqn:gamma1} is derived from a Lindbladian generator, ${\cal L}(\rho) = \sum_k \lambda_k (P_k \rho P_k - \rho)$.  This allows us to express $\gamma$ in terms of the spectrum of ${\cal L}$.  That is
\begin{align}
    {\rm Tr}({\cal L}) = -\sum \lambda_k d^2 \implies \gamma = e^{\frac{-2 {\rm Tr}({\cal L}) }{ d^2}}
\end{align}
Through the exponential map we arrive at
\begin{align}
    \gamma = {\rm det}(\Lambda)^{\frac{-2}{d^2}}
    \label{eq:gamma-det}
\end{align}
That is gamma is related to the geometric mean of the eigenvalues of $\Lambda^2$. 

\subsection{depolarizing channels}
For a depolarizing channel, the spectrum of $\Lambda$ is a single $1$ and $(d^2-1)$, $(1-\alpha)$'s.  In this case
\begin{align}
    \frac{{\rm Tr}(\Lambda)}{d^2} & = \frac{1}{d^2} + \left(1 -  \frac{1}{d^2}  \right) (1-\alpha)\\
    {\rm det}(\Lambda)^{\frac{1}{d^2}} &= (1-\alpha)^{1-\frac{1}{d^2}}
\end{align}
In the limit of large $d$ these both converge to $(1-\alpha)$.  In terms of this depolarizing parameter $\alpha$ we have,
\begin{align}
\varepsilon(\Lambda) &\approx \frac{\alpha}{1+\frac{1}{d}},\\
\| \Lambda \|_\diamond &\approx 2 \alpha ,\\
\gamma &\approx (1-\alpha)^{-2}. 
\end{align}
Alternatively, we can express the gate fidelity and diamond norm in terms of gamma as
\begin{align}
\varepsilon(\Lambda) &\approx \frac{1-1/\sqrt{\gamma}}{1+\frac{1}{d}},\\
\| \Lambda \|_\diamond &\approx 2 (1-1/\sqrt{\gamma}).
\end{align}

\subsection{The small error limit}
Let's assume $\Lambda$ is very close to the identity, i.e., the spectrum contains terms $1-\epsilon_j$.  Let's define $\bar{\epsilon} \equiv \frac{1}{d^2} \sum_j \epsilon_j$.
\begin{align}
    \frac{{\rm Tr}(\Lambda)}{d^2} & = 1-\bar{\epsilon}\\
    {\rm det}(\Lambda)^{\frac{1}{d^2}} &= \prod_j (1-\epsilon_j)^{\frac{1}{d^2}} = 1- \bar{\epsilon} + {\cal O}(\epsilon^2)
\end{align}
Once again we are in the limit where the arithmetic and geometric means are the same, which again yields
\begin{align}
\varepsilon(\Lambda) &\approx \frac{1-1/\sqrt{\gamma}}{1+\frac{1}{d}},\\
\| \Lambda \|_\diamond &\approx 2 (1-1/\sqrt{\gamma}).
\end{align}

\subsection{Bounds}

The process fidelity of a superoperator is the arithmetic mean of its eigenvalues.
On the other hand, Eqn.~\ref{eq:gamma-det} established that $\gamma^{-\frac12}$ is equal to the geometric mean of the eigenvalues.
To make $F_p$ and $\gamma$ more easily comparible this section chooses to work in terms of $\gamma^{-\frac12}$

We start with Theorem~\ref{thm:bounds} which provides upper and lower bounds for $\gamma^{-\frac12}$ in terms of $F_p$.
Although the lower bound appears complicated, it is extremely close to $\sqrt{2F_p-1}$ on all of $[\frac{1}{2}, 1]$, which can therefore be used as a proxy for most practical purposes.
Especially note that both the upper and lower bounds are independent of the dimension $d=2^n$.
Following the theorem, we provide natural families of channels that saturate the upper bound, and nearly saturate the lower bound.
For high fidelity layers, say above $F_p=0.9$, it will be seen that $F_p\approx \gamma^{-1/2}$.

\begin{theorem}
Suppose $\Lambda$ is a CPTP Pauli channel with a process fidelity $F_p=\tr\Lambda/d^2$, and  $\gamma = \det(\Lambda)^{-2/d^2}$.
Then it holds that 
\begin{align}
    F_p - 1 + 2\lambda_0(1-F_p)+(2F_p-1)^{\lambda_0} \leq \gamma^{-\frac{1}{2}} \leq F_p
\end{align}
where 
\begin{align}
    \lambda_0
    = \frac{\log(2-2F_p)-\log(-\log(2F_p-1))}{\log(2F_p-1)}.
\end{align}
\label{thm:bounds}
\end{theorem}
\begin{proof}
The upper bound on $\gamma^{-\frac12}$ follows directly from a standard application of Jensen's inequality; the geometric mean of positive numbers cannot exceed their arithmetic mean.

To show the lower bound, first observe that all of the Pauli fidelities $f_a$ of $\Lambda$ lie in the interval $[2F_p-1, 1]$.
This was shown in Ref.~\cite{erhard:2019}, but we repeat the brief argument here for completeness.
We can express the Pauli fidelity $f_a$ in terms of the Kraus probabilities as
\begin{align}
    f_a &= \sum_b (-1)^{\langle a, b\rangle} p_b
    = \sum_{b: \langle a, b\rangle=0} p_b - \sum_{b: \langle a, b\rangle\neq 0} p_b \nonumber \\
    &= 2 \sum_{b: \langle a, b\rangle=0} p_b - 1,
\end{align}
where we have used the CPTP condition $\sum_b p_b=1$ to write the sum of those $p_b$ where $b$ does not commute with $a$ as one minus the sum of those that do.
Now clearly $\sum_{b: \langle a, b\rangle=0} p_b \geq p_I$, and moreover $p_I=4^{-n} \sum_a (-1)^{\langle I, a\rangle}f_a=F_p$, hence
\begin{align}
    f_a \geq 2 F_p -1.
\end{align}

We can now apply Lemma~\ref{lem:gap} (below) with $c=F_p-1$ and $d=1$ to get the stated inequality.
\end{proof}

The geometric mean and arithmetic mean agree exactly when their arguments are equal, which means that $F_p=\gamma^{-\frac12}$ exactly when all of the Pauli fidelities are equal, which, for TP channels, only happens with the identity channel.
However, globally depolarizing channels are the next best thing as all values but one are equal.
For a globally depolarizing channel with non-trivial Pauli fidelities $\alpha$, we have
\begin{align}
    F_p = \frac{1 + (4^n-1)\alpha}{4^n}
    \quad\text{and}\quad
    \gamma^{-\frac12} = \alpha^{(4^n-1)/4^n},
\end{align}
which are very close to equal even for moderate $n$; see the upper curve in Fig.~\ref{fig:gamma-bounds}.

The next family of channels we consider are tensor products of two-qubit depolarizing channels, each with strength $\alpha$.
Assuming $n$ is even and we have the tensor product of $n/2$ such channels, we get
\begin{align}
    F_p = \sum_{k=0}^{n/2} \alpha^k 15^k \binom{n/2}{k}/4^n
    \quad\text{and}\quad
    \gamma^{-\frac12} = \alpha^{15n/32},
\end{align}
which, as with global depolarizing channels (see Fig.~\ref{fig:gamma-bounds}), are very close to equal even for moderate $n$.

The previous two families of channels have had Pauli fidelites that are tightly concentrated.
To saturate the lower bound of Theorem~\ref{thm:bounds}, we will need to instead choose a channel that maximizes the variance of the Pauli fidelities.
As seen in the proof of Lemma~\ref{lem:gap}, this is done by making the Pauli fidelities strongly bimodal, concentrating roughly half of them at 1, and the other half at another value less than 1.
This can be done by choosing a $\Lambda$ to have a Kraus map with only one non-trivial Pauli, 
\begin{align}
    \Lambda(\rho) = p\rho + (1-p) P\rho P.
\end{align}
In this case, $f_a=1$ whenever $a$ commutes with $P$, but $2p-1$ otherwise.
In this case we get
\begin{align}
    F_p = p
    \quad\text{and}\quad
    \gamma^{-\frac12}= \sqrt{2p-1},
\end{align}
where we note the independence of $n$.
A physically relevant example of such a noise model is one where a single subsystem has a high error that dominates every other error, for example, by taking $P=XIII\cdots I$, where the first qubit is faulty.
The geometric mean (viz. $\gamma$) is good at capturing this outlier, but the arithmetic mean (viz. $F_p$) is not.

\begin{figure}[ht!]
    \centering
    \includegraphics[width=\columnwidth]{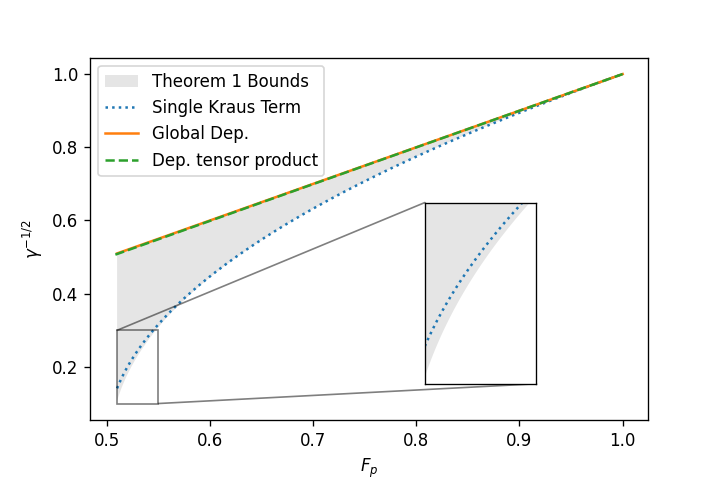}
    \caption{For any fixed process fidelity on the x-axis, the grey region above it represents the range of values of $\gamma^{-\frac12}$ that are physically consistent by some Pauli channel with that process fidelity.
    The three curves depict three families of Pauli channels that saturate the bounds of the grey region, two on the top and one on the bottom.
    The global depolarizing curve (orange) is shown for $n=10$, but would look identical for any other $n$ not too close to $1$.
    Likewise, the 5-tensor product of 2-qubit depolarizing (dashed green) curve also corresponds to $n=10$, but would look similar for more qubits.
    The lower bound is (just about; see inset) saturated by channels for which there is only one non-trivial term in the Kraus representation, representative of noise models where there is a single outlying subsystem, such as $\Lambda(\rho)=p\rho + (1-p)XIIII\rho XIIII$ (dotted blue).
    }
    \label{fig:gamma-bounds}







\end{figure}

\begin{lemma}
Suppose that $0<c<d$ are real numbers and fix a positive integer $N$. Define $a:[0,\infty)^N\rightarrow \mathbb{R}$ by $a(x)=\sum_i x_i/N$ and $g:[0,\infty)^N\rightarrow \mathbb{R}$ by $g(x)=\prod_i x_i^{1/N}$.
Then restricting to the hyperrectangular region $D=[c,d]^N$, we have
\begin{align}
    \max_{x\in D} (a(x) - g(x))
    \leq f(\lambda_0)
\end{align}
where $f(\lambda)=\lambda c + (1-\lambda)d + c^\lambda d^{1-\lambda}$
and 
\begin{align}
    \lambda_0 = \frac{\log(\log(d/c))-\log((d-c)/d)}{\log(d/c)}.
\end{align}
\label{lem:gap}
\end{lemma}
\begin{proof}
Since $a$ is linear and $g$ is concave, $a-g$ must be convex.
Therefore, the maximum of $a-g$ is acheived on the extreme points of the convex set $D$, which are given by
 $E=\{a,b\}^N$.
 That is, we have $\max_{x\in D} (a-g)(x)=\max_{e\in E} (a-g)(e)$.
 
Now, for any $e\in E$, there exists some $0\leq m \leq N$ such that $(a-g)(e)=(mc+(N-m)d)/N + (c^m d^{N-m})^{1/N}$.
Therefore, $\max_{e\in E} (a-g)(e) \leq \max_{0\leq \lambda \leq 1} f(\lambda)$.
Standard calculus show that $f$ is concave on $[0,1]$ and acheives a maximum value at $\lambda_0$, which proves the inequality of this lemma.
\end{proof}

\end{appendix}

\end{document}